\documentclass[letterpaper]{llncs}

\usepackage{soulutf8} 

\usepackage{cite} 
\usepackage{amssymb}
\usepackage{amsmath}
\usepackage{textcomp}
\usepackage{array}
\usepackage{color}
\usepackage{comment}
\usepackage{tikz}
\usetikzlibrary{arrows,automata}

\newcommand{\bra}[1]{\langle #1|}
\newcommand{\ket}[1]{|#1\rangle}
\newcommand{\braket}[2]{\langle #1|#2\rangle}
\newcommand{\cent}[0]{\mbox{\textcent}}
\newcommand{\dollar}[0]{\$}

\newtheorem{fact}{Fact}

\newcommand{\Hil}{\mathcal{H}}
\newcommand{\Hiln}{\mathcal{H}^n}
\newcommand{\Real}{\mathbb{R}}
\newcommand{\Realn}{\mathbb{R}^n}

\newcommand{\kettilde}[1]{\ket{\widetilde{#1}}}

\newcommand{\mymatrix}[2]{\left( \begin{array}{#1} #2\end{array} \right)}

\newcommand{\myvector}[1]{\mymatrix{c}{#1}}
\newcommand{\myrvector}[1]{\mymatrix{r}{#1}}

\newcommand{\cvector}[1]{\left( \begin{array}{c} #1 \end{array} \right)}

\newcommand{\tildesigma}{\widetilde{\Sigma}}
\newcommand{\tildew}{\tilde{w}}

\newcommand{\reg}{\mathsf{REG}}
\newcommand{\stoc}{\mathsf{SL}}
\newcommand{\affine}{\mathsf{AfL}}
\newcommand{\naffine}{\mathsf{NAfL}}
\newcommand{\baffine}{\mathsf{BAfL}}
\newcommand{\posonebaffine}{\mathsf{BAfL^0}}
\newcommand{\negonebaffine}{\mathsf{BAfL^1}}
\newcommand{\exactaffine}{\mathsf{EAfL}}
\newcommand{\excstoc}{\mathsf{SL^{\neq}}}

\newcommand{\coexcstoc}{\mathsf{SL^{=}}}

\newcommand{\excaffine}{\mathsf{AfL^{\neq}}}
\newcommand{\coexcaffine}{\mathsf{AfL^{=}}}
\newcommand{\nqal}{\mathsf{NQAL}}

\newcommand{\lapins}{\mathtt{LAPIN\check{S}}}

\newcommand{\TSigma}{\widetilde{\Sigma}}

\usepackage[]{datetime}

\title{Affine computation and affine automaton}

\author{Alejandro D\'{i}az-Caro\inst{1}$^,$\thanks{D\'{i}az-Caro was partially supported by STIC-AmSud project 16STIC04 FoQCoSS.} \and Abuzer Yakary{\i}lmaz\inst{2}$^,$\thanks{Yakary{\i}lmaz was partially supported by CAPES with grant 88881.030338/2013-01 and some parts of the work were done while Yakary{\i}lmaz was visiting Buenos Aires in July 2015 to give a lecture at ECI2015 (Escuela de Ciencias Inform\'{a}ticas 2015, Departamento de Computaci\'{o}n, Facultad de Ciencias Exactas y Naturales, Universidad de Buenos Aires), partially supported by CELFI, Ministerio de Ciencia, Tecnolog\'{i}a e Innovaci\'{o}n Productiva.}}
\institute{
  Universidad Nacional de Quilmes\\
  Roque S\'aenz Pe\~na 352, B1876BXD Bernal, Buenos Aires, Argentina\\
  \email{alejandro.diaz-caro@unq.edu.ar}
  \and 
  National Laboratory for Scientific Computing\\
  Petr\'{o}polis, RJ, 25651-075, Brazil\\
  \email{abuzer@lncc.br}
}

\authorrunning{A. D\'{i}az-Caro \and A. Yakary{\i}lmaz} 

\begin{document}

\maketitle

\begin{abstract}
	We introduce a quantum-like classical computational concept, called affine computation, as a generalization of probabilistic computation. After giving the basics of  affine computation, we define affine finite automata (AfA) and compare it with quantum and probabilistic finite automata (QFA and PFA, respectively) with respect to three basic language recognition modes. We show that, in the cases of  bounded and unbounded error, AfAs are more powerful than QFAs and PFAs, and, in the case of nondeterministic computation, AfAs are more powerful than PFAs but equivalent to QFAs. Moreover, we show that exclusive affine languages form a superset of exclusive quantum and stochastic languages. 
\end{abstract}

\section{Introduction}

Using negative amplitudes, allowing interference between states and configurations, is one of the fundamental properties of quantum computation that does not exist in classical computation. Therefore, it is interesting to define a quantum-like classical system allowing to use negative values. However, both quantum and probabilistic systems are linear and it seems not possible to define a classical linear computational systems using negative values (see also the discussions regarding fantasy quantum mechanics in \cite{Aar05}). On the other hand, it is possible to define such a system \textit{almost linearly}, as we do in this paper.

A probabilistic state is a $l_1$-norm 1 vector defined on non-negative real numbers, also called a stochastic vector. A probabilistic operator is a linear operator mapping probabilistic states to probabilistic states, which is also called a stochastic matrix. Equivalently, a matrix is stochastic if each of its columns is a probabilistic state. Similarly, a quantum state is a $ l_2 $-norm 1 vector defined over complex numbers. A quantum operator is a linear operator mapping quantum states to quantum states, which is also called a unitary matrix. Equivalently, a matrix is unitary if each of its columns (also each row) is a quantum state. 

Our aim is to define a new system that (1) is a generalization of probabilistic system, (2) can have negative values, (3) evolves linearly, and (4) is defined in a simple way like the probabilistic and quantum systems. When working on non-negative real numbers, $ l_1 $-norm is the same as the summation of all entries. So, by replacing ``$ l_1 $-norm 1'' condition with ``summation to 1'' condition, we can obtain a new system that allows negative values in its states. The linear operators preserving the summation of vectors are barycentric-preserving, also called \textit{affine transformations}, which can give the name of our new system: \textit{affine system}. Thus, the state of an affine system is called an \textit{affine state}, the entries of which sum to 1. Moreover, a matrix is an affine transformation if each of its columns is an affine state. It is clear that an affine system is a probabilistic one if negative values are not used. Thus, the new affine system satisfies all four conditions above.

The only renaming detail is how to get information from the system. For this purpose, we define an operator similar to the measurement operator in quantum computation that projects the computation into the computational basis. It is intuitive that the ``weights'' of negative and positive values should be same if their magnitudes are the same. Moreover, each state should be observed with the probability calculated based on the value of its magnitude. Therefore, we normalize each magnitude (since the summation of all magnitudes can be bigger than 1) and each normalized magnitude gives us the probability of ``observing'' the corresponding state. We call this operator as \textit{weighting operator}.

In the paper, we give the basics of affine systems in detail and start to investigate affine computation by defining the affine finite automaton (AfA) (due to the simplicity of automata models). Then, we compare it with probabilistic finite automata (PFAs) and quantum finite automata (QFAs) with respect to the basic language recognition modes. We show that, in the cases of  bounded and unbounded error, AfAs are more powerful than QFAs and PFAs, and, in the case of nondeterministic computation, AfAs are more powerful than PFAs but equivalent to QFAs. Moreover, we show that exclusive affine languages form a superset of exclusive quantum and stochastic languages. Our results are also the evidence that although an AfA has a finite number of basis states, it can store more information. This is why we use small ``f'' in the abbreviation of AfA.


Throughout the paper, we focus on the finite dimensional systems. In Section \ref{sec:pqsystem}, we give the basics of probabilistic and quantum systems. In Section \ref{sec:affine-system}, we describe the basics of affine systems. Then, we give the definitions of classical and quantum finite automata in Section \ref{classical-automata}. The definition of the affine finite automaton is given in Section \ref{sec:afa}. Our results are given in Section \ref{sec:main}. 
We close the paper with Section \ref{sec:conclusion}.

\section{Probabilistic and quantum systems}
\label{sec:pqsystem}

A \textit{probabilistic system} has a finite number of states, say $ E = \{ e_1,\ldots,e_n \} $ $ (n>0) $, called \textit{deterministic states} of the system. At any moment, the system can be in a probabilistic distribution of these states:
\[
  v = \myvector{ p_1 \\ p_2 \\ \vdots \\ p_n },
\]
where $ p_j $ represents the probability of system being in state $ e_j $ ($ 1 \leq j \leq n $). Here $ v $ is called a \textit{probabilistic state}, which is a stochastic (column) vector, i.e.
\[
  0 \leq p_i \leq 1 \mbox{ and } \sum_{i=1}^n p_i = 1.
\]
It is clear that $ v $ is a vector in $ \Realn $ and all the deterministic states form the standard basis of $ \Realn $. Moreover, all the probabilistic states form a simplex in $ \Realn $, represented by linear equation $ x_1+x_2+\cdots+x_n =1 $ whose variables satisfy $ 0 \leq x_j \leq 1 $ ($ 1 \leq j \leq n$).

The system evolves from a probabilistic state to another one by a linear operator:
\[
  v' = Av,
\]
where $ A $ is an $ n \times n $ matrix and $ A[k,j] $ represents the probability of going from $ e_j $ to $ e_k $ ($ 1 \leq j,k \leq n $). Since $ v' $ is a probabilistic state and so a stochastic vector, $ A $ is a (left) stochastic matrix, each column of which is a stochastic vector. Assume that the system is in $ v_0 $ at the beginning, and $ A_t $ is the probabilistic operator at the $t$-th time step ($ t = 1,2,\ldots $). Then, the evolution of the system is as follows:
\[
  v_t = A_t A_{t-1} \cdots A_1 v_0.
\]
At the $t$-th step, the probability of observing the $j$-th state is $ v_t[j] $.

A quantum system is a \textit{non-trivial linear} generalization of a probabilistic one, which forms a Hilbert space (a complex vector space with inner product). A basis of the Hilbert space, say $\Hiln$, can be seen as the set of ``deterministic states'' of the system. Unless otherwise is specified, the standard basis is used: $ B = \{ \ket{q_1}, \ldots, \ket{q_n} \} $, where each $ \ket{q_j} $ is a zero vector except the $j$-th entry, which is 1. Remark that $ \Hiln = span\{ \ket{q_1},\ldots,\ket{q_n} \} $. At any moment, the system can be in a linear combination of basis states:
\[
  \ket{v} = \myvector{ \alpha_1 \\ \alpha_2 \\ \vdots \\ \alpha_n } \in \Hiln , ~~~ 
\]
where $ \alpha_j \in \mathbb{C} $ is called the \textit{amplitude} of the system being in state $ \ket{q_j} $. Moreover, the value $ | \alpha_j |^2 $ represents the \textit{probability} of the system being in state $ \ket{q_j} $. We call $ \ket{v} $ the \textit{(pure) quantum state} of the system, which is a norm-1 (column) vector:
\[
  \sqrt{\braket{v}{v}} = 1 \Leftrightarrow \braket{v}{v} = 1 \Leftrightarrow  \sum_{j=1}^n |\alpha_j|^2 = 1. 
\]
Remark that all the quantum states from a sphere in $ \mathbb{C}^n $, i.e. $ x_1^2 + x_2^2 + \cdots + x_n^2 = 1 $.

Similar to the probabilistic case, the system evolves from a quantum state to another one by a linear operator:
\[
  \ket{v'} = U \ket{v},
\]
where $ U $ is an $ n \times n $ matrix and $ U[k,j] $ represents the transition amplitude of going from $ \ket{q_j} $ to $ \ket{q_k} $ ($ 1 \leq j,k \leq n $). Since $ \ket{v'} $ is a quantum state and so a norm-1 vector, $ U $ is a unitary matrix, the columns/rows of which form an orthonormal set. Moreover, $ U^{-1} = U^{\dagger} $. 

To retrieve information from a quantum system, we apply measurement operators. In its simplest form, when in quantum state $ \ket{v} $, we can make a measurement in the computation basis and then we can observe $ \ket{q_j} $ with probability $ p_j= |\alpha_j|^2 $ and so the new state becomes $ \ket{q_j} $  (if $ p_i >0 $). We can also split the set $ B $ into $ m $ disjoint subsets: $ B = B_1 \cup \cdots \cup B_m $ and $ B_j \cap B_k = \emptyset $ for $ 1 \leq j \neq k \leq m $. Based on this classification, $\Hiln$ is split into $ m $ pairwise orthogonal subspaces: $ \Hiln = \Hil_1 \oplus \cdots \oplus \Hil_m $ where $ \Hil_j = span\{ \ket{q} \mid \ket{q} \in B_j \} $. We can design a projective measurement operator $ P $ to force the system to be observed in one of these subspaces, i.e.
\[
  P = \left\{ P_1,\ldots,P_m \mid P_j = \sum_{\ket{q} \in B_j} \ket{q}\bra{q} \mbox{ and } 1 \leq j \leq m \right\},
\]
where $ P_j $ is a zero-one projective matrix that projects any quantum state to $ \Hiln_j $. More formally,
\[
  \ket{v} = \kettilde{v_1} \oplus \cdots \oplus \kettilde{v_m},~~ \kettilde{v_j} = P_j \ket{v} \in \Hiln_j.
\]
Each $ P_j $ is called a \textit{projective operator} and the index is called a \textit{measurement outcome}. Then, the probability of observing the outcome ``$j$'' is calculated as
\[
  p_j = \braket{\widetilde{v_j}}{\widetilde{v_j}}.
\]
If it is observed ($ p_j>0 $), then the new state is obtained by normalizing $ \kettilde{v_j} $, which is called \textit{unnormalized (quantum) state},
\[
  \ket{v_j} = \frac{\kettilde{v_j}}{\sqrt{p_j}}.
\]

From a mathematical point of view, any quantum system defined on $ \Hiln $ can be simulated by a quantum system straightforwardly defined on $ \Real^{2n} $ (e.g. \cite{MC00}). Therefore, we can say that the main distinguishing property of quantum systems is using negative amplitudes rather than using complex numbers.

After making projective measurements, for example, the quantum system can be in a mixture of pure quantum states, i.e.
\[
  \left\{ (p_j,\ket{v_j}) \mid 1 \leq j \leq m, \sum_{j=1}^m p_j = 1 \right\}.
\]
We can represent such a mixture as a single mathematical object called \textit{density matrix}, an $(n \times n)$-dimensional matrix:
\[
  \rho =  \sum_{j=1}^m p_j \ket{v_j}\bra{v_j},
\]
which is called the \textit{mixed state} of the system. A nice property of $ \rho $ is that the $ k $-th diagonal entry represents the probability of the system of being in the state $ \ket{q_k} $, i.e. $ Tr(\rho) = 1 $.

It is clear that unitary operators are not the generalizations of stochastic operators. However, by interacting a quantum system with an auxiliary system, more general quantum operators can be applied on the main quantum system. They are called \textit{superoperators}.\footnote{A superoperator can also be obtained by applying a series of unitary and measurements operators where the next unitary operator is selected with respect to the last measurement outcome.} Formally, a superoperator $ \cal E $ is composed by a finite number of operation elements $ \{ E_j \mid 1 \leq j \leq m \} $, where $ m>0 $, satisfying that
\[
  \sum_{j=1}^m E_j^\dagger E_j = I.
\]
When $ \cal E $ is applied to the mixed state $ \rho $, the new mixed state is obtained as
\[
  \rho' = \mathcal{E}(\rho) = \sum_{j=1}^m E_j \rho E_j^\dagger.
\]
In fact, a superoperator always includes a measurement and the indices of operation elements can be seen as the outcomes of the measurement(s). When $ \cal E $ is applied to pure state $ \ket{v} $, we can obtain up to $ m $ new pure states. The probability of observing the outcome of ``$j$'', say $p_j$, calculated as
\[
	p_j = \braket{\widetilde{v_j}}{\widetilde{v_j}},~~~\ket{\widetilde{v_j}} = E_j \ket{v},
\]
where $ \ket{\widetilde{v_j}}  $ is called an unnormalized state vector if it is not a zero vector. If the outcome ``$j$'' is observed ($p_j>0$), then the new state becomes,
\[
	\ket{v_j} = \frac{ \ket{\widetilde{v_j}}}{\sqrt{p_j}}.
\]
Remark that using unnormalized state vectors sometimes make the calculations easier since the probabilities can be calculated directly from them.

If we apply the projective measurement $ P = \{ P_j \mid 1  \leq j \leq m \} $ to the mixed state $ \rho $, where $ m>0 $, the probability of observing the outcome $ j $, say $p_j$, and the new state, say $ \rho_j $, is calculated as follows:
\[
  \widetilde{\rho_j} = P_j \rho P_j, ~~~ p_j = Tr(\widetilde{\rho_j}), ~~ \mbox{and} ~~ \rho_j = \frac{\widetilde{\rho_j}}{p_j} (\mbox{if } p_j>0).
\]

The reader may ask how a quantum system can be a linear generalization of a probabilistic system. We omit the details here but any probabilistic operator can be implemented by a superoperator. Moreover, a mixed-state can be represented as a single column vector, and each superoperator can be represented as a single matrix. Then, all computations can be represented linearly. We refer the reader to \cite{Wat09A,YS11A,SayY14} for the details.

\section{Affine systems}
\label{sec:affine-system}

Inspired from quantum systems, we define the finite-dimensional affine system (AfS) as a \textit{non-linear} generalization of a probabilistic system by allowing to use negative ``probabilities''. Let $ E = \{ e_1,\ldots,e_n \} $ be the set of basis states, which are the deterministic states of an $ n $-dimensional probabilistic system. Any affine state is a linear combination of $ E $
\[
  v = \cvector{a_1 \\ a_2 \\ \vdots \\ a_n}
\]
such that each entry can be an arbitrary real number but the summation of all entries must be 1:
\[
  \sum_{i=1}^n a_i = 1.
\]
So, any probabilistic state, a stochastic column vector, is an affine state. However, on contrary to a probabilistic state, an affine state can contain negative values. Moreover, all the affine states form a surface in $\Realn $, i.e. $ x_1 + x_2 + \cdots + x_n =  1 $.

Both, probabilistic and quantum states, form finite objects (simplex and sphere, respectively). For example, in $\Real^2$, all the probabilistic states form the line $ x+y=1 $ on $ (\Real^+ \cup \{0\})^2 $ with length $ \sqrt{2} $ and all the quantum states form the unit circle with length $ 2 \pi $. On the other hand, affine states form infinite objects (plane). In $ \Real^2 $, all the affine states form the infinite line $ x+y=1 $. Therefore, it seems that, with the same dimension, affine systems can store more information. In this paper, we provide some evidences to this interpretation. On the other hand, affine systems might not be comparable with quantum systems due to the fact of forming different geometrical objects (e.g. line versus circle).  


Any affine transformation is a linear operator, that is, a mapping between affine states. We can easily show that any matrix is an affine operator if and only if for each column, the summation of all entries is equal to 1. The evolution of the system is as follows: when in affine state $v$, the new affine state $ v' $ is obtained by
\[
  v' = A v,
\]
where $ A $ is the affine transformation such that $ A[j,k] $ represents the transition value from $ e_k $ to $ e_j $.

In quantum computation, the sign of the amplitudes does not matter when making a measurement. We follow the same idea for affine systems. More precisely, the \textit{magnitude} of an affine state is the $l_1$-norm of the state:
\[
  |v| = | a_1 | + |a_2| + \cdots + |a_n| \geq 1.
\]
Then, we can say that the probability (\textit{weight}) of observing the $ j $-th state is
\[
  \frac{|a_j|}{|v|},
\]
where $ 1 \leq j \leq n $.
To retrieve this information, we use an operator (possible non-linear) called \textit{weighting operator}, which can be seen as a counterpart of the measurements in the computational basis for quantum systems. Therefore, we can make a weighting in the basis $ E $ and the system collapses into a single deterministic state.

One may ask whether we can use a weighting operator similar to a projective measurement. Assume that the system is in the following affine state
\[
  v = \mymatrix{r}{1\\-1\\1}
\]
and we make weighting based on the separation $ \{ e_1 \} $ and $ \{ e_2,e_3 \} $. Then, we can observe the system in the first state with weight $ \frac{1}{3} $ and in the second and third states with weight $ \frac{2}{3} $. But, in the latter case, the new state is not an affine state since the summation of entries will always be zero whatever normalization factor is used. Therefore, once we make a weighting, the system must collapse to a single state. On the other hand, one may still define an \textit{affine system with extended weighting} by allowing this kind of weighting with the assumption that if the new state  has a zero summation, then the system terminates, i.e. no further evolution can occur. Such kind of assumptions may be used cleverly to gain some computational power.


One may also define an affine state as a $ l_1 $-norm 1 vector on the real numbers and require that each new state is normalized after each linear affine operator. A straightforward calculation shows that the weighting results will be exactly the same as the previous definition, so both systems are equivalent. However, this time the overall evolution operator, a linear affine operator followed by normalization, is not linear.
With respect to this new definition, say \textit{normalized affine systems}, all the affine states form finite objects: $ |x_1|+|x_2|+\cdots+|x_n| = 1 $. It is, for example, a square on $ \Real^2 $: $ |x|+|y| =1 $. One could see this square as an approximation of the unit circle but remark that we cannot use unitary operators as affine operators directly. On the other hand, we may define a more general model by replacing linear affine operators with arbitrary linear operators. We call this system \textit{general affine systems} or \textit{general normalized affine systems}. In this paper, we focus only on the standard definition where the states are vectors with a barycentric sum to $1$, and the transformations are affine operators preserving such barycenters.

\section{Classical and quantum automata}
\label{classical-automata}

Unless otherwise specified, we denote the input alphabet as $ \Sigma $, not containing the left end-marker $\cent$ and the right end-marker $ \dollar $. The set of all the strings generated on $ \Sigma $ is denoted by $ \Sigma^* $. We define $ \tildesigma = \Sigma \cup \{ \cent,\dollar \} $ and $ \tildew = \cent w \dollar $ for any string $w \in \Sigma^* $. For any given string $ w \in \Sigma^* $, $|w|$ is the length of the string, $|w|_\sigma$ is the number of occurrences of the symbol $\sigma$ in $w$, and $w_j$ is the $j$-th symbol of $w$. 

For a given machine/automaton $ M $, $ f_{M}(w) $ denotes the accepting probability (value) of $M$ on the string $ w $.

A probabilistic finite automaton (PFA) \cite{Rab63} $ P $ is 5-tuple
\[
  P = (E,\Sigma,\{ A_{\sigma} \mid \sigma \in \tildesigma \},e_s,E_a),
\]
where $ E $ is the set of deterministic states, $e_s \in E$ is the starting state, $ E_a \subseteq E $ is the set of accepting state(s), and $ A_{\sigma} $ is the stochastic transition matrix for the symbol $ \sigma \in \widetilde{\Sigma} $. Let $ w \in \Sigma^* $ be the given input. The input is read as $ \tildew $ from left to right, symbol by symbol. After reading the $ j $-th symbol, the probabilistic state is
\[
  v_j = A_{\tildew_j} v_{j-1} = A_{\tildew_j} A_{\tildew_{j-1}} \cdots A_{\tildew_1} v_0,   
\] 
where $ v_0 = e_s $ and $ 1 \leq j \leq |\tildew| $. The final state is denoted $ v_f = v_{|\tildew|} $. The accepting probability of $ P $ on $w$ is calculated as
\[
  f_P(w) = \sum_{e_k \in E_a} v_f[k].
\]

A quantum finite automaton (QFA) \cite{AY15} $ M $ is a 5-tuple
\[
	M = (Q,\Sigma,\{ \mathcal{E}_\sigma \mid \sigma \in \TSigma \},q_s,Q_a),
\]
where $ Q $ is the set of basis states, $ \mathcal{E}_\sigma $ is the transition superoperator  for symbol $ \sigma $, $ q_s $ is the starting state, and $ Q_a \subseteq Q $ is the set of accepting states. For a given input $ w \in \Sigma^* $, the computation of $ M $ on $ w $ is traced as
\[
	\rho_j = \mathcal{E}_{\tilde{w}_j} ( \rho_{j-1} ),
\]
where $ \rho_0 = \ket{q_s}\bra{q_s} $ and $ 1 \leq j \leq |\tilde{w}| $. The final state is denoted $ \rho_f = \rho_{|\tilde{w}|} $. The accepting probability of $ M $ on $ w $ is calculated as
\[
	f_M(w) =  \sum_{q_j \in A_a} \rho_f[j,j].
\]

If we restrict the entries of the transitions matrices of a PFA to zeros and ones, we obtain a deterministic finite automaton (DFA). A DFA is always in a single state during the computation and the input is accepted if and only if the computation ends in an accepting state. A language is said to be recognized by a DFA (then called regular \cite{RS59}) if and only if any member of the language is accepted by the DFA. The class of regular languages are denoted by $\reg$. 

Let $ \lambda \in \lbrack 0,1) $ be a real number. A language $ L $ is said to be recognized by a PFA $P$ with cutpoint $ \lambda $ if and only if
\[
  L = \{ w \in \Sigma^* \mid f_P(w) > \lambda \}.
\]
Any language recognized by a PFA with a cutpoint is called stochastic language \cite{Rab63} and the class of stochastic languages are denoted by $ \stoc $, a superset of $\reg$. A language is said to be recognized by a PFA $  P$ with unbounded error if  $ L $ or the complement of $ L $ is recognized by  $P$ with cutpoint \cite{YS11A}. (Remark that it is still not known whether $ \stoc $ is closed under complement operation.) 

As a special case,  if $ \lambda = 0 $, the PFA is also called a nondeterministic finite automaton (NFA). Any language recognized by a NFA is also regular. 

A language $ L $ is said to be recognized by $ P $ with isolated cutpoint $ \lambda $ if and only if there exists a positive real number $ \delta $ such that 
\begin{itemize}
  \item $ f_P(w) \geq \lambda + \delta $ for any $ w \in L $ and
  \item $ f_P(w) \leq \lambda - \delta $ for any $ w \notin L $.
\end{itemize}
When the cutpoint is required to be isolated, PFAs are not more powerful than DFAs: Any language recognized by a PFA with isolated cutpoint is regular \cite{Rab63}. 

Recognition with isolated cutpoint can also be formulated as recognition with bounded error. Let $ \epsilon \in \lbrack 0,\frac{1}{2}) $. A language $ L $ is said to be recognized by a PFA $ P $ with error bound $ \epsilon $ if and only if
\begin{itemize}
  \item $ f_P(w) \geq 1 - \epsilon $ for any $ w \in L $ and
  \item $ f_P(w) \leq \epsilon $ for any $ w \notin L $.
\end{itemize}
As a further restriction of bounded error, if $ f_P(w) = 1 $ for any $ w \in L $, then it is called negative one-sided error bound, and, if $ f_P(w) = 0 $ for any $ w \notin L $, then it is called positive one-sided error bound. If the error bound is not specified, it is said that $ L $ is recognized by $ P $ with \lbrack negative/positive one-sided\rbrack\ bounded error. 

A language $ L $ is called exclusive stochastic language \cite{Paz71} if and only if there exists a PFA $P$  and a cutpoint $ \lambda \in [0,1]  $ such that 
\[
  L = \{ w \in \Sigma^* \mid f_P(w) \neq \lambda \}.
\]
The class of exclusive stochastic languages is denoted by $ \excstoc $. Its complement class is denoted by $ \coexcstoc $ ($ L \in \excstoc \leftrightarrow \overline{L}\in \coexcstoc $). Note that for any language in $ \excstoc $ we can pick any cutpoint between 0 and 1 but not 0 or 1 since when fixing the cutpoint to 0 or 1, we can recognize only regular languages. Note that both $\excstoc$  and $ \coexcstoc$ are supersets of $ \reg $ (but it is still open whether $ \reg $ is a proper subset of $ \excstoc \cap \coexcstoc $). 

In the case of QFAs, they recognize all and only regular languages with bounded-error \cite{LQZLWM12} and stochastic languages with cutpoint \cite{YS09C,YS11A}. However, their nondeterministic versions (NQFAs) are more powerful: $\nqal$, the class of languages defined by NQFAs (QFAs with cutpoint 0), is identical to $ \excstoc $ \cite{YS10A}. 

\section{Affine finite automaton}
\label{sec:afa}

Now we define the affine finite automaton (AfA). An AfA $ M $ is a 5-tuple
\[
  M = (E,\Sigma,\{ A_{\sigma} \mid \sigma \in \tildesigma\},e_s,E_a),
\]
where all the components are the same as that of PFA except that $ A_{\sigma} $ is an affine transformation matrix. Let $ w \in \Sigma^* $ be the given input. After reading the whole input, a weighting operator is applied and the weights of the accepting states determine the accepting probability of $M$ on $w$, i.e.
\[
  f_M(w) = \sum_{e_k \in E_a} \frac{| v_f[k] |}{|v_f|} \in [0,1].
\]

The languages recognized by AfAs are defined similarly to PFAs and QFAs. Any language recognized by an AfA with cutpoint is called \textit{affine language}. The class of affine languages is denoted $ \affine $. Any language recognized by an AfA with cutpoint 0 (called nondeterministic AfA (NAfA)) is called \textit{nondeterministic affine language}. The related class is denoted $ \naffine $. A language is called exclusive affine language if and only if there exists an AfA $M$  and a cutpoint $ \lambda \in [0,1]  $ such that 
\[
  L = \{ w \in \Sigma^* \mid f_M(w) \neq \lambda \}.
\]
The class of exclusive affine languages is denoted by $ \excaffine $ and its complement class is denoted by $ \coexcaffine $. Any language recognized by an AfA with bounded error is called \textit{bounded affine language}. The related class is denoted $ \baffine $. If it is a positive one-sided error (all non-members are accepted with value 0), then the related class is denoted $ \posonebaffine$, and, if it is a negative one (all members are accepted with value 1), then the related class is denoted $ \negonebaffine $. Note that if $ L \in \posonebaffine $, then $ \overline{L} \in \negonebaffine $, and vice versa. Any language recognized by an AfA with zero-error is called \textit{exact affine language} and the related class is denoted $ \exactaffine $.

\section{Main results}
\label{sec:main}

We present our results under five subsections. 

\subsection{Bounded-error languages}

We start with a 2-state AfA, say $M_1$, for the language $ \mathtt{EQ} = \{ w \in \{a,b\}^* \mid |w|_a = |w|_b \} $. Let $ E = \{e_1,e_2\} $ be the set of states, where $ e_1 $ is the initial and only accepting state. None of the end-markers is used (or the related operators are the identity). At the beginning, the initial affine state is
\[
  v_0 = \myvector{1 \\ 0}.
\] 
When reading symbols $ a $ and $ b $, the following operators are applied:
\[
  A_a = \mymatrix{rr}{2 & 0 \\ -1 & 1} ~~~ A_b = \mymatrix{rr}{\frac{1}{2} & 0 \\ \frac{1}{2} & 1},
\]
respectively. 
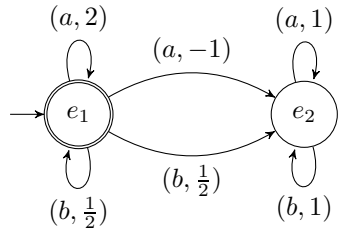
\begin{center}
  \begin{tikzpicture}[initial text=,>=stealth',auto,node distance=3cm]
    \node[state,initial,accepting] (e1) {$e_1$};
    \node[state] (e2) [right of=e1] {$e_2$};
    \path[->] 
    (e1) edge [loop above] node[above]{$(a,2)$} (e1)
    (e1) edge [loop below] node[below]{$(b,\frac12)$} (e1)
    (e1) edge [bend left] node[above] {$(a,-1)$} (e2)
    (e1) edge [bend right] node[below] {$(b,\frac12)$} (e2)
    (e2) edge [loop above] node[above] {$(a,1)$} (e2)
    (e2) edge [loop below] node[below] {$(b,1)$} (e2);
  \end{tikzpicture}
\end{center}

Then, the value of the first entry of the affine state is multiplied by 2 for each $ a $ and by $ \frac{1}{2} $ for each $ b $, and so, the second entry takes the value of ``1 minus the value of the first entry'', i.e. if $ M $ reads $ m $ $a$s and $ n $ $b$s, then the new affine state is
\[
  \myvector{ 2^{m-n} \\ 1 - 2^{m-n} }.
\]
That is, for any member, the final affine state is 
$
v_f = \myvector{1 \\ 0}
$
and so the input is accepted with value 1. For any non-member, the final state can be one of the followings
\[
  \cdots, \myrvector{8 \\ -7}, \myrvector{4 \\ -3}, \myrvector{2 \\ -1}, \myvector{\frac{1}{2} \\ \frac{1}{2}}, \myvector{\frac{1}{4} \\ \frac{3}{4}}, \myvector{\frac{1}{8} \\ \frac{7}{8}}, \cdots.
\]
Thus, the maximum accepting value is obtained when $ v_f =  \myrvector{2 \\ -1} $, which gives the accepting value $ \frac{|2|}{|2|+|-1|} = \frac{2}{3} $. Therefore, we can say that the language $ \tt EQ $ can be recognized by the AfA $M_1$ with isolated cutpoint $ \frac{5}{6} $ (the isolation gap is $\frac{1}{6}$). Since it is a nonregular language, we can follow that AfAs can recognize more languages than PFAs and QFAs with isolated cutpoints (bounded error).

By using 3 states, we can also design an AfA $ M_2(x) $ recognizing $\tt EQ$ with better error bounds, where $ x \geq 1 $:
\[
  M_2(x) = \{ \{e_1,e_2,e_3\},\{a,b\}, \{ A_a, A_b \}, e_1,\{e_1\} \}, \mbox{ where}
\] 
\[
  A_a = \mymatrix{rrr}{1 & 0 & 0 \\ x & 1 & 0 \\ -x & 0 & 1}
  ~~\mbox{and}~~
  A_b = \mymatrix{rcc}{1 & 0 & 0 \\ -x & 1 & 0 \\ x & 0 & 1}.
\]
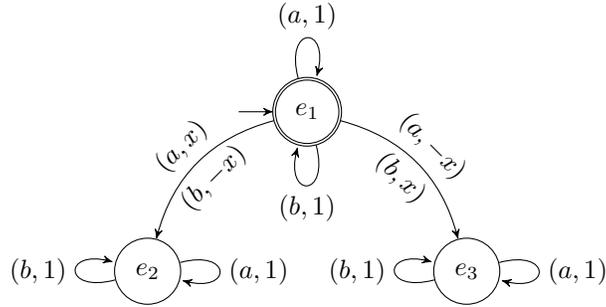
\begin{center}
  \begin{tikzpicture}[initial text=,>=stealth',auto,node distance=3cm]
    \node[state,initial,accepting] (e1) {$e_1$};
    \node[state] (e2) [below left of=e1] {$e_2$};
    \node[state] (e3) [below right of=e1] {$e_3$};
    \path[->] 
    (e1) edge [loop above] node[above]{$(a,1)$} (e1)
    (e1) edge [loop below] node[below]{$(b,1)$} (e1)
    (e1) edge [bend right] node[sloped,above]{$(a,x)$} node[sloped,below]{$(b,-x)$}(e2)
    (e1) edge [bend left] node[sloped,above]{$(a,-x)$} node[sloped,below]{$(b,x)$}(e3)
    (e2) edge [loop right] node[right]{$(a,1)$} (e2)
    (e2) edge [loop left] node[left]{$(b,1)$} (e2)
    (e3) edge [loop right] node[right]{$(a,1)$} (e3)
    (e3) edge [loop left] node[left]{$(b,1)$} (e3)
    ;
  \end{tikzpicture}
\end{center}

The initial affine state is $ v_0 = (1,0,0) $ and after reading $ m $ $a$s and $n$ $b$s, the affine state will be
\[
  \mymatrix{c}{1\\(m-n)x \\ (n-m)x}.
\]
Then, the accepting value will be $1$ if $m=n$, and,
$
\frac{1}{2x|m-n|+1}
$
if $m\neq n$. Notice that it is at most $ \frac{1}{2x+1} $ if $ m \neq n $. Thus, by picking larger $ x $, we can get smaller error bound. 

\begin{theorem}
  $ \reg \subsetneq \negonebaffine $ and $ \reg \subsetneq \posonebaffine \subseteq \naffine $.
\end{theorem}

The knowledable readers can notice that in the algorithm $ M_2(x) $, we actually implement a blind counter \cite{Gre78}\footnote{A counter is blind if its status (whether its value is zero or not) cannot be accessible during the computation. A multi-blind-counter finite automaton is an automaton having $k>0$ blind counter(s) such that in each transition it can update the value(s) of its counter(s) but never access the status of any counter. Moreover, an input can be accepted by such automaton only if the value of every counter is zero at the end of the computation.}. Therefore, by using more states, we can implement more than one blind counter.

\begin{corollary}
  Any language recognized by a deterministic multi-blind-counter automaton is in $ \negonebaffine $.
\end{corollary}

Since AfA is a generalization of PFA, we can also obtain the following result.

\begin{theorem}
  Any language recognized by a probabilistic multi-blind-counter automaton with bounded-error is in $ \baffine $.
\end{theorem}

\subsection{Changing cutpoint}

Before giving the other results, we present a few technical results regarding the cases where the choice of a cutpoint is essential. 

For given automata $M_1$ and $M_2$, we say that $ L(M_1,\lambda_1) $ is equivalent to $ L(M_2,\lambda_2) $, denoted
\[
  L(M_1,\lambda_1) \equiv L(M_2,\lambda_2)
\]
if for any input $ w \in \Sigma^* $,
\begin{enumerate}
  \item $ f_{M_1}(w) > \lambda_1 \rightarrow f_{M_2}(w) > \lambda_2 $,
  \item $ f_{M_1}(w) = \lambda_1 \rightarrow f_{M_2}(w) = \lambda_2 $, and
  \item $ f_{M_1}(w) < \lambda_1 \rightarrow f_{M_2}(w) < \lambda_2 $.
\end{enumerate}

Up to date, it is a folkloric result that for any given $n$-state PFA or QFA, say $ M_1 $, and a cutpoint $ \lambda_1 \in [0,1] $, we can define another PFA or QFA $ M_2 $ with $ (n+1) $ states and cutpoint $ \lambda_2 \in (0,1) $ such that
\[
  L(M_1,\lambda_1) \equiv L(M_2,\lambda_2).
\]

For AfAs, we can obtain the same results with different state overheads.

\begin{theorem}
  Let $ M_1 $ be an $n$-state AfA and $ \lambda_1 \in (0,1) $, then for any $ \lambda_2 \in (0,1) $, we can define another AfA $ M_2 $ with $ (n+2) $ states such that 
  \[
    L(M_1,\lambda_1) \equiv L(M_2,\lambda_2).
  \]
\end{theorem}
\begin{proof}
  The AfA $ M_2 $ is obtained from $ M_1 $ by adding two more states and making certain modifications. Let $ w \in \Sigma^* $ be an input and $ v_f $ be the final vector of $ M_1 $ after reading $ \tilde{w} $. We can represent $v_f$ as the summation of two orthogonal vectors: $ v_f = v^a_f + v^r_f $, where $ v^a_f $ is the projection of $ v_f $ on the space spanned by the accepting states, i.e. $ v^a_f $ is obtained from $ v_f $ by setting entries of non-accepting states to zeros, and $ v^r_f = v_f - v^a_f  $. We define $ |A| $ and $ |R| $ as the $ l_1 $-norms of $ v^a_f $ and $ v^r_f $ respectively. Remark that $ f_M(w) = \frac{|A|}{|A|+|R|} $ and, 
  \begin{itemize}
    \item if $ f_M(w) > \lambda_1 $, then $ \frac{|A|}{|R|} > \frac{\lambda_1}{1-\lambda_1}  $,
    \item if $ f_M(w) = \lambda_1 $, then $ \frac{|A|}{|R|} = \frac{\lambda_1}{1-\lambda_1}  $, and
    \item if $ f_M(w) < \lambda_1 $, then $ \frac{|A|}{|R|} < \frac{\lambda_1}{1-\lambda_1}  $.
  \end{itemize}

  The AfA $ M_2 $ follows the same computation of $ M_1 $ except that it applies an additional affine transformation on the right end-marker $ \dollar $, say $ A'_{\dollar} $: The final state of $ M_2 $ is 
  \[
    u_f = A'_\dollar \myvector{v_f \\ 0 \\ 0},
  \]
  where the effect of $ A'_\dollar $ is as follows:
  \begin{itemize}
    \item each value of the accepting state(s) in $ v_f $ is multiplied by $ \frac{\lambda_2}{\lambda_1} $,
    \item each value of the non-accepting state(s) in $ v_f $ is multiplied by $ \frac{1-\lambda_2}{1-\lambda_1} $,
    \item the value of the $ (n+1) $-th state in $ u_f $ is set to $ \lambda_2(1-T) $, and
    \item the value of the $ (n+2) $-th state in $ u_f $ is set to $ (1-\lambda_2)(1-T) $,
  \end{itemize}
  where $ T  $ is the summation of all entries except the last two in $ u_f $: $ T = \sum_{i=1}^n u_f[i] $. It is easy to verify that $ u_f $ is an affine state:
  \[
    \sum_{i=1}^{n+2} u_f[i] = T + \lambda_2 \left(1-T\right) + (1-\lambda_2)\left(1-T\right) = 1. 
  \]
  The accepting states of $M_2$ is formed by the accepting state(s) of $ M_1 $ and the $ (n+1) $-th state. For $ u_f $, we similarly define $ u^a_f $ and $ u^r_f $. Then, we define $ |A'| $ and $ |R'| $ as the $ l_1 $-norms of $ u^a_f $ and $ u^r_f $:
  \[
    |A'| =  \frac{\lambda_2}{\lambda_1} |A| + \lambda_2 \left| 1-T \right|
    \mbox{ and }
    |R'| = \frac{1-\lambda_2}{1-\lambda_1} |R| + (1-\lambda_2) \left| 1-T \right|.
  \]

  Now, we are ready to verify our construction:
  \begin{itemize}
    \item If $ f_{M_1}(w) = \lambda_1 $, we have $ \frac{|A|}{|R|} = \frac{\lambda_1}{1-\lambda_1}  $. Then, we calculate the following ratio:
      \[
	\frac{|A'|}{|R'|} = \frac{  \frac{\lambda_2}{\lambda_1} |A| + \lambda_2 \left| 1-T \right| }{ \frac{1-\lambda_2}{1-\lambda_1} |R| + (1-\lambda_2) \left| 1-T \right| }
      \]
      We can replace $ |A| $ with $ |R| \frac{\lambda_1}{1-\lambda_1}  $ in the above formula:
      \[
	\frac{|A'|}{|R'|} = \frac{ \lambda_2 \left( \frac{|R|}{1-\lambda_1 } + \left| 1-T \right|  \right) }{ (1 - \lambda_2) \left( \frac{|R|}{1-\lambda_1 } + \left| 1-T \right|  \right) } = \frac{ \lambda_2 }{1- \lambda_2}.
      \]
      That means $ f_{M_2}(w) = \lambda_2 $.
    \item If $ f_{M_1}(w) > \lambda_1 $, we have $ \frac{|A|}{|R|} > \frac{\lambda_1}{1-\lambda_1}  $.  This time, we replace $ |A| $ with 
      \[
	|R| \frac{\lambda_1}{1-\lambda_1} + \delta \lambda_2
      \] 
      for some  $ \delta>0 $ in the ratio of $ \frac{|A'|}{|R'|} $. Then, we get
      \[
	\frac{|A'|}{|R'|} = \frac{ \lambda_2 \left( \frac{|R|}{1-\lambda_1 } + \delta + \left| 1-T\right|  \right) }{ (1 - \lambda_2) \left( \frac{|R|}{1-\lambda_1 } + \left| 1- T \right|  \right) } > \frac{ \lambda_2 }{1- \lambda_2}.
      \]
      That means $ f_{M_2}(w) > \lambda_2 $.
    \item For the case, $ f_{M_1}(w) < \lambda_1 $, we can obtain that $ f_{M_2}(w) < \lambda_2 $ in the same way by replacing $ -\delta \lambda_2 $ with $ \delta \lambda_2 $.  
  \end{itemize}
  Therefore, $ L(M_1,\lambda_1) \equiv L(M_2,\lambda_2) $. 	
  \qed\end{proof}

The construction in the above proof does not work when $ \lambda_1 = 0 $ or $ \lambda_1 = 1 $. Therefore, we give another proof.

\begin{theorem}
  Let $ M_1 $ be an $n$-state AfA with $ k < n $ non-accepting state(s), then for any $ \lambda \in (0,1) $, we can define another AfA $ M_2 $ with $ (n+k) $ states such that 
  \[
    L(M_1,0) \equiv L(M_2,\lambda).
  \]
\end{theorem}
\begin{proof}
  Suppose that the first $ k $ states of $ M_1 $ are non-accepting states. The AfA $ M_2 $ is obtained by modifying $ M_1 $. The first $ k $ states of $ M_2 $ are non-accepting and all the others are accepting states. Until reading the right end-marker, $ M_2 $ trace the computation of $ M_1 $ exactly with the same states. On the right end-marker, $ M_2 $ applies first the operator of $ M_1 $ and then applies an additional one that (1) multiplies the value of each non-accepting $ i $-th state with $ (1-\lambda) $ ($ 1 \leq i \leq k $) and (2) transfers the values of the non-accepting states to the additional states after multiplying $ \lambda $, i.e. the value of the $ (n+i) $-th state is set to the multiplication of the value of the $ i $-th state and the value of $ \lambda $.

  Therefore, it is easy to derive that for any $ w \in \Sigma^* $, if $ f_{M_1}(w) = 0 $, then $ f_{M_2}(w) = \lambda  $; and, if $ f_{M_1}(w) > 0 $, then $ f_{M_2}(w) > \lambda  $. 
  \qed\end{proof}

\begin{theorem}
  Let $ M_1 $ be an $n$-state AfA with $ k < n $ accepting state(s), then for any $ \lambda \in (0,1) $, we can define another AfA $ M_2 $ with $ (n+k) $ states such that 
  \[
    L(M_1,1) \equiv L(M_2,\lambda).
  \]
\end{theorem}
\begin{proof}
  The proof is the analogous to the proof of the previous theorem, except that we focus on accepting state(s) instead of non-accepting state(s).
  \qed\end{proof}

\subsection{Cutpoint languages}

Lapin\v{s} \cite{Lap74} showed that the language $ \lapins = \{ a^{m}b^{n}c^{p} \mid m^{4} > n^{2} > p > 0 \} $ is nonstochastic and it is not in $ \stoc $. It is clear that the following language is also nonstochastic:
\[
  \lapins' = \{ w \in \{ a,b,c \}^* \mid |w|^4_a > |w|^2_b > |w|_c \}
\]
or equivalently
\[
  \lapins' = \{ w \in \{ a,b,c \}^* \mid |w|^2_a > |w|_b \mbox{ and } |w|^2_b > |w|_c \}.
\]

\begin{theorem}
  \label{thm:lapins}
  The language $\lapins'$ is recognized by an AfA with cutpoint $ \frac{1}{2} $. 
\end{theorem}
\begin{proof} 

  We start with a basic observation about AfAs. Let $ \myvector{m\\1-m} $ be an affine state, where $ m $ is an integer. Then:
  \begin{itemize}
    \item If $ m >0 $, $ |m| > |1-m| = m-1 $ since $m-1$ is closer to 0. 
    \item If $ m \leq 0 $, $ |m| = -m < |1-m| = -m+1  $ since $ -m $ is closer to 0.
  \end{itemize}    
  So, if the first state is an accepting state and the second one is not, then we can determine whether $ m>0 $ with cutpoint $ \frac{1}{2} $, which can be algorithmically useful.

  For a given input $ w \in \{ a,b,c \}^* $, we can easily encode $ |w|_a $, $ |w|_b $, and $ |w|_c $ into the values of some states. Our aim is to determine $ |w|^2_a > |w|_b $ and $ |w|^2_b > |w|_c  $. Even though PFAs and QFAs can make the similar encodings, they can make only a single comparison. Here, we show that AfAs can make both compressions. First we present some encoding integer matrices. If we apply matrix $\mymatrix{cc}{1 & 0 \\1 & 1}  $  $ m $ times to $ \myvector{1\\0} $, the value of the second entry becomes $ m $:
  \[
    \mymatrix{cc}{1 & 0 \\1 & 1} \myvector{1\\0} = \myvector{1\\1}
    \mbox{ and }
    \mymatrix{cc}{1 & 0 \\1 & 1} \myvector{1\\x} = \myvector{1\\x+1} 
  \]
  \[
    \Rightarrow
    \mymatrix{cc}{1 & 0 \\1 & 1}^m \myvector{1\\0} = \myvector{1\\m}.
  \]
  For obtaining $ m^2 $, we can use the following initial state and matrix:
  \[
    \mymatrix{ccc}{1 & 0 & 0 \\ 2 & 1 & 0 \\ 0 & 1 & 1} \myvector{1\\ 0\\ 0} = \myvector{1\\1\\0}
    \mbox{ and }
    \mymatrix{ccc}{1 & 0 & 0 \\ 2 & 1 & 0 \\ 0 & 1 & 1} \myvector{1\\2x-1\\(x-1)^2} = \myvector{1\\2x+1\\x^2}
  \] 
  \[
    \Rightarrow
    \mymatrix{ccc}{1 & 0 & 0 \\ 2 & 1 & 0 \\ 0 & 1 & 1}^m \myvector{1\\0\\0} = \myvector{1\\2m-1\\m^2}.
  \]
  We can easily embed such matrices into affine operators (by using some additional states) and then we can obtain the value like $ |w|_a $ and $ |w|_a^2 $ as the values of some states. If required, the appropriate initial states can be prepared on the left end-marker. Moreover, on the right end-marker, we can make some basic arithmetic operations by using a combination of more than one affine operators. Furthermore, we can easily tensor two AfA and obtain a single AfA that indeed can simulate both machines in parallel.

  Let $ |w|_a = x $, $ |w|_b = y $, and $ |w|_c =z $, and, $ M_1 $ and $ M_2 $ be two AfAs that respectively have the following final states after reading $ w $: 
  \[
    v_f(M_1) = \mymatrix{c}{x^2 \\ y \\ 1-x^2-y \\ 0 \\ \vdots \\ 0} ~~~
    v_f(M_2) = \mymatrix{c}{y^2-z \\ 1-y^2+z \\ 0 \\ \vdots \\ 0}.
  \]
  If we tensor both machines and apply an affine transformation to arrange the values in a certain way, the final state will be
  \[
    \arraycolsep=1pt\def\arraystretch{2}
    v_f = \mymatrix{c}{
      x^2 ( y^2 - z ) \\
      x^2 ( 1 - y^2 + z ) \\
      y \\
      \dfrac{1-T}{2}  \\ 
      \dfrac{1-T}{2} \\
      0 \\ \vdots \\ 0
    },
  \]
  where $ T $ is the summation of the first three entries. We select the first and fourth states as accepting states. Then, the difference between the accepting and the remaining values is
  \[
    \Delta = x^2 ( | y^2-z | - | 1-y^2 +z | ) - y.
  \]
  Remark that $ \delta = | y^2-z | - | 1-y^2 +z | $ is either 1 or $ -1 $.   
  \begin{itemize}
    \item If $ w $ is a member, then  $ \Delta = x^2 (1) - y $, which is greater than 0.
    \item If $ w $ is not a member, then we have different cases. 
      \begin{itemize}
	\item $ x^2 \leq y $: $ \Delta $ will be either $ x^2 - y $ or $ -x^2 - y $ and in both case it is equal to zero or less than zero.
	\item $ x^2 > y $ but $ y^2 \leq z $: $ \Delta $ will be $ -x^2 - y $ and so less than zero.
      \end{itemize}
  \end{itemize}
  Thus, the final AfA can recognize $\lapins'$ with cutpoint $ \frac{1}{2} $.
  \qed\end{proof}

Since AfAs can recognize a nonstochastic language with cutpoint, they are more powerful than PFAs and QFAs with cutpoint (and also with unbounded-error).

\begin{corollary}	
  $ \stoc \subsetneq \affine $.
\end{corollary}

\subsection{Nondeterministic languages}

Now, we show that NAfAs are equivalent to NQFAs.

\begin{lemma}
  $ \excstoc \subseteq \naffine $.
\end{lemma}
\begin{proof}
  Let $L$ be a language in $ \excstoc $. Then, there exists an $n$-state PFA $ P $ such that
  \[
    L = \{ w \in \Sigma^* \mid f_P(w) \neq \frac{1}{2} \}, 
  \]
  where $n>0 $. Let $ A_\dollar $ be the transition matrix for the right end-marker and $ v_f(w) $ be the final probabilistic state for the input $w \in \Sigma^*$. We can trivially design  a probabilistic transition matrix $ A'_\dollar $ such that the first and second entries of the probabilistic state
  \[
    v'_f(w) = A'_\dollar v_f(w)
  \]
  are $ 1 - f_P(w) $ and $ f_P(w) $, respectively, and the others are zeros. Let $ A''_\dollar $ be the following affine operator:
  \[
    \mymatrix{rr|c}{1 & -1 & \mbox{\textbf{0}} \\ 0 & 2 & \mbox{\textbf{0}} \\ \hline \mbox{\textbf{0}} & \mbox{\textbf{0}} & \mbox{\textbf{I}}}.
  \]
  Then, the first and second entries of $ v''_f = A''_\dollar v'_f $ are $ 1 - 2 f_P(w) \mbox{ and } 2f_P(w) $, respectively, and the others are zeros. So, based on $ P $, we can design an AfA $ M $ by making at most two modifications: (i) the single accepting state of $M$ is the first one and (ii) the affine operator for the right end-marker is $ A''_\dollar A'_\dollar A_\dollar  $. Then, if $ f_P(w) = \frac{1}{2} $ if and only if $ f_M(w) = 0 $. That is, $ L \in \naffine $.
  \qed
\end{proof}

\begin{lemma}
  $ \naffine \subseteq \nqal $.
\end{lemma}
\begin{proof}
  Let $ L \in \naffine $. Then, there exists an AfA $ M = (E,\Sigma,\{ A_\sigma \mid \sigma \in \TSigma,e_s,E_a \}) $ such that $ w \in L $ if and only if $ f_M(w) > 0 $ for any input $ w \in \Sigma^* $. Now, we design a nondeterministic QFA $ M' = \{ E \cup F ,\Sigma,\{ \mathcal{E}_\sigma \mid \sigma \in \Sigma \}, e_s, E_a \} $ for language $L$, where $ F $ is a set of finite states and $ E \cap F = \emptyset $. 

  We provide a simulation of $ M $ by $M'$. The idea is to trace the computation of $ M $ through a single pure state. Let $ w \in \Sigma^* $ be the input string. The initial affine state is $ v_0 = e_s $ and the initial quantum state is $ \ket{v_0} = \ket{e_s} $. Assume that each superoperator has $ k>0 $ operation elements.

  A superoperator can map a pure state to more than one pure state. Therefore, the computation of $ M' $ can be also traced/shown as a tree, say $ T_w $. We build the tree level by level. The root is the initial state. For the first level,  we apply $ \mathcal{E}_{\cent} $ to the initial state and obtain $ k $ vectors:
  \[
    \ket{\widetilde{v_{(j)}}} = {E}_{\cent,j} \ket{v_0}, ~~~ 1 \leq j \leq k,
  \] 
  some of which are unnormalized pure states and maybe the others are zero vectors. We connect all these vectors to the root. For the second level, we apply $ \mathcal{E}_{\tildew_2} $ to each vectors on the first level. Although it is clear that zero vectors can always be mapped to zero vectors, we keep them for simplicity. From the node corresponding $ \ket{\widetilde{v_{(j)}}} $, we obtain the following children:
  \[
    \ket{\widetilde{v_{(j,j')}}} = E_{\tildew_2,j'} \ket{\widetilde{v_{(j)}}}, ~~~ 1 \leq j' \leq k.
  \]
  We continue in this way (by increasing the indices of vectors by one in each level) and at the end, we obtain $ k^{|\tildew|} $ vectors at the leafs, some of which are unnormalized pure states. The indices of the vectors at the leafs are from $ (1,\ldots,1) $ to $ (k,\ldots,k) $. Remark that $ \ket{\widetilde{v_{(1,\ldots,1)}}} $ is calculated as
  \[
    \ket{\widetilde{v_{(1,\ldots,1)}}} = E_{\tildew_{|\tilde{w}|},1}	E_{\tildew_{|\tilde{w}|-1},1} \cdots E_{\tildew_{1},1} \ket{v_0},
  \]
  where all the operation elements are the ones having index of 1. Remark that if $ \alpha $ is a value of an accepting state in one of these pure states, then its contribution to the total accepting probability will be $ |\alpha|^2 $. 

  This tree facilities to describe our simulation. Each superoperator $ \mathcal{E}_{\sigma} = \{ E_{\sigma,1}, \ldots, E_{\sigma,k} \} $ is defined based on $ A_{\sigma} $. Among the others, $ E_{\sigma,1} $ is the special one that keeps the transitions of $ A_{\sigma} $ and all the others exist for making $ \mathcal{E}_{\sigma} $ a valid operator. The details of $ E_{\sigma,1} $ and the other operation elements of $ \mathcal{E}_{\sigma} $ are as follows:
  \[
    E_{\sigma,1} = \frac{1}{l_\sigma} \mymatrix{c|c}{A_\sigma & 0 \\ \hline 0 & I} 
    \mbox{ and }
    E_{\sigma,j} = \frac{1}{l_\sigma} \mymatrix{c|c}{0 & 0 \\ \hline * & *},~~~(2 \leq j \leq k)
  \]
  where $ l_\sigma \geq 1 $ is a normalization factor and the parts denoted by ``$*$'' can be arbitrary filled to make $ \mathcal{E}_{\sigma} $ a valid operator, which can be also formulated as follows:
  the columns of the following matrix must form an orthonormal set \cite{Yak12B}.
  \[
    \mymatrix{c}{ 
      \frac{1}{l_\sigma} \mymatrix{c|c}{A_\sigma & 0 \\ \hline 0 & I} 
      \\
      \frac{1}{l_\sigma} \mymatrix{c|c}{ \mspace{7mu} 0 ~ & 0 \\ \hline * & *}
      \\
      ~~~~~~ \vdots
      \\
      \frac{1}{l_\sigma} \mymatrix{c|c}{ \mspace{7mu} 0 ~ & 0 \\ \hline * & *}
    }
  \]
  Note that there have already been some methods to fill the parts denoted by ``$*$'' in a straightforward way \cite{YS10A,YS11A}.  

  The Hilbert space of $ M' $ can be decomposed into two orthogonal subspaces: $ \mathcal{H}_e = span\{ \ket{e} \mid e \in E \} $ and $ \mathcal{H}_f = span\{ \ket{f} \mid f \in F \} $. So, any pure state $ \ket{v} $ can be decomposed as $ \ket{v} = \ket{v_e} \oplus \ket{v_f} $, where $ \ket{v_e} \in \mathcal{H}_e $ and $ \ket{v_f} \in \mathcal{H}_f $. It is clear that any $ E_{\sigma,1} $ ($ \sigma \in \TSigma $) keeps the vector inside of the subspaces: $ E_{\sigma,1} : \mathcal{H}_e \rightarrow \mathcal{H}_e $ and $ E_{\sigma,1} : \mathcal{H}_f \rightarrow \mathcal{H}_f $. Then, $ E_{\sigma,1} $ maps $ \ket{v} = \ket{v_e} \oplus \ket{v_f} $ to $ \frac{1}{l_\sigma} A_\sigma \ket{v_e} \oplus \frac{1}{l_\sigma} \ket{v_f}  $. Therefore, when $ E_{\sigma,1} $ is applied, the part of computation in $ \mathcal{H}_f $ never affects the part in $ \mathcal{H}_e $.

  All the other operational elements map any vector inside $ \mathcal{H}_f $ and so they never affect the part in $ \mathcal{H}_e $. Remark that any pure state lies in $ \mathcal{H}_f $ never produce an accepting probability since the set of accepting states are a subset of $ E $.

  Now, we have enough details to show why our simulation works. When considering all leaves of $ T_w $, only $ \ket{\widetilde{v_{(1,\ldots,1)}}} $ lies in $ \mathcal{H}_e $ and all the others lie in $ \mathcal{H}_{f} $. Then, the accepting probability can be produced only from $ \ket{\widetilde{v_{(1,\ldots,1)}}} $, the value of which can be straightforwardly calculated as
  \[
    \ket{\widetilde{v_{(1,\ldots,1)}}} =  \frac{1}{l_w} (v_f,*,\ldots,*) , ~~~ l_w = \prod_{j=1}^{|\tildew|} l_{\tildew_j},
  \]
  where ``$*$'' are some values of the states in $ F $. It is clear that $ f_M(w) = 0 $ if and only if $ f_{M'}(w) = 0 $.

  Remark that each superoperator can have a different number of operation elements and this does not change our simulation. Moreover, the size of $ F $ can be arbitrary. 
  If it is small, then we need to use more operation elements and if it is big enough, then we can use less operation elements. 
  \qed
\end{proof}
\begin{theorem}
  $ \naffine = \nqal $.
\end{theorem}
\begin{proof}
  The equality follows from the fact that  $ \excstoc = \nqal $ \cite{YS10A} and the previous two lemmas:
  $
  \nqal \subseteq \excstoc \subseteq \naffine \subseteq \nqal
  $
  \qed
\end{proof}
\subsection{Exclusive stochastic languages}

For PFAs, exclusive languages ($ \excstoc $) are larger than nondeterministic (cutpoint 0) languages ($ \reg $). On the other hand, for QFAs, exclusive languages and nondeterministic languages are identical ($ \nqal = \excstoc $). For AfAs, we show that nondeterministic languages ($ \naffine $) are identical $ \excstoc = \nqal $. The interesting question here is whether AfAs show a similar behaviour to PFAs or QFAs when comparing exclusive and nondeterministic languages. Now, we show that, for AfAs, exclusive languages ($ \excaffine $) are larger than nondeterministic languages ($ \naffine $) similar to PFAs. For this purpose, we use the complement classes $ \coexcstoc $ and $ \coexcaffine $.

\newcommand{\abseq}{\mathtt{ABS \mbox{-} EQ}}
The language $ \abseq $ is defined on $ \{a,b\} $ such that $ w \in \abseq $ if and only if 
\begin{equation}
  \label{eq:abseq}
  | m-n  | + | m-4n | = | m-2n | + | m-3n |,
\end{equation}
where $ |w|_a = m  $ and $ |w|_b=n $.

It is clear that $ \abseq $ contains every string if we do not use absolute values in Equation \ref{eq:abseq}. Therefore, the interesting part of $\abseq$ is the absolute values, which leads us to understand the power of the weighting operator. Remark that if $ m \geq 4n $ and $ m \leq n $, Equation \ref{eq:abseq} is trivially satisfied, and, if $ m \in (n,4n) $, then Equation \ref{eq:abseq} is never satisfied.

First, we show that $ \abseq $ is not in $ \coexcstoc $ by using the following fact \cite{Die71}.

\begin{fact}
  Let $ L $ be a language in $ \coexcstoc $. Hence, there exists an $ n $-state PFA $ P $ such that $ w \in L $ if and only if $ f_P(w) = \frac{1}{2} $. Then, for any $ x,y,z \in \Sigma^* $, 
  \[
    \mbox{if } xz, xyz, xy^2z,\ldots,xy^{n-1}z \in L, \mbox{ we also have } x y^* z \in L. 
  \]
\end{fact}

\begin{theorem}
  $ \abseq \notin \coexcstoc $.
\end{theorem}
\begin{proof}
  Suppose that there exists $ n $-state PFA $ P $ as described in the above fact for the language $ \abseq $ for $ n>1 $. We pick $ x = a^{8n}b $, $ z = b^{n} $, and $ y = b $. Then, we can have the following list (remember that as long as $ |w|_a - 4 |w|_b \geq 0 $, Equation \ref{eq:abseq} is trivially satisfied):
  \[
    \begin{array}{lllll}
      w=xy & = a^{8n}b^{n+1} & \mbox{is in } \abseq \mbox{ since } & |w|_a - 4 |w|_b = 4n-4 & \geq 0 \\
      w=xyz & = a^{8n}b^{n+2} & \mbox{is in } \abseq \mbox{ since } & |w|_a - 4 |w|_b = 4n - 8 & \geq 0 \\
      w=xy^jz & = a^{8n}b^{n+j+1} & \mbox{is in } \abseq \mbox{ since } & |w|_a - 4 |w|_b = 4n-4j-4 & \geq 0  \\
      w=xy^{n-1}z & = a^{8n}b^{2n} & \mbox{is in } \abseq \mbox{ since } & |w|_a - 4 |w|_b = 0 & \geq 0 \\
    \end{array}
  \]
  Due to the above fact, $ xy^nz = a^{8n}b^{2n+1} $ must be in $ \abseq $ but Equation 1 cannot be satisfied for $ xy^nz $:
  \begin{eqnarray*}
    | 8n - (2n+1) | + | 8n - 4(2n+1) | & = & | 8n - 2(2n+1) | + | 8n - 3(2n+1) | \\
    (6n - 1) + 4 & = & (4n-2) + (2n-3) \\
    6n + 3 & = & 6n - 5 
  \end{eqnarray*}	
  Therefore, $ xy^nz  $ is not in $ \abseq $ and so $ \abseq $ is not a member of $ \coexcstoc $.	
  \qed\end{proof}

Now, we present our AfA algorithm for $ \abseq $ which only calculates the values inside the absolute values in Equation \ref{eq:abseq} and the desired decision is given by the weighting operator. 

\begin{theorem}
  $ \abseq $ is in $ \coexcaffine $.
\end{theorem}
\begin{proof}
  We design a $ 6 $-state AfA $ M $ for $ \abseq $. The initial state is $ (1~~0~~0~~0~~0~~0)^T $. The operator on $ \cent $ is identity. Let $ w \in \{a,b\}^* $ be the given input and $ |w|_a = m $ and $ | w |_b = n $. After reading $ w $, the values of $ m $ and $ n $ are stored in the second and third states:
  \[
    v_{|\cent w|} = \myvector{1-m-n \\ m \\ n \\ 0 \\ 0 \\ 0}
  \]
  The updates for three states for symbols $ a $ and $ b $ are given below. The value of the second (first) entry is increased (decreased) by 1 when reading an $ a $:
  \[
    \myvector{-m'-n' \\ m'+1 \\ n'} =  \mymatrix{rrr}{0  & ~-1 & ~-1 \\ 1 & 2 & 1 \\ 0 & 0 & 1 } \myvector{ 1-m'-n' \\ m' \\ n' }.
  \]
  The value of the third (first) entry is increased (decreased) by 1 when reading a $ b $:
  \[
    \myvector{-m'-n' \\ m' \\ n'+1} =  \mymatrix{rrr}{0  & ~-1 & ~-1 \\ 0 & 1 & 0 \\ 1 & 1 & 2 } \myvector{ 1-m'-n' \\ m' \\ n' }.
  \]
  On the right end-marker, we apply the following affine operator:
  \[
    v_f =
    \mymatrix{c}{ m-n \\ m-2n \\ m-3n \\ m-4n \\ \frac{1-T}{2} \\ \frac{1-T}{2} }		
    =
    \mymatrix{rrrrrr}{ 
      0 & ~1 & ~-1 & ~1 & ~0 & ~0 \\
      0 & 1 & -2 & 0 & 1 & 0 \\
      0 & 1 & -3 & 0 & 0 & 1 \\
      0 & 1 & -4 & 0 & 0 & 0 \\
      \frac{1}{2} & ~-\frac{3}{2} & \frac{11}{2} & 0 & 0 & 0 \\
      \frac{1}{2} & ~-\frac{3}{2} & \frac{11}{2} & 0 & 0 & 0
    }
    \myvector{1-m-n \\ m \\ n \\ 0 \\ 0 \\ 0},
  \]
  where $ T = 4m-10n $, the summation of the first four entries in $ v_f $. Let the first, fourth, and fifth states be the accepting ones. Then $ f_M (w) = \frac{1}{2}$ if and only if 
  \[
    | m-n | + | m-4n | + \left| \frac{1-T}{2} \right| = | m-2n | + | m-3n | + \left| \frac{1-T}{2} \right|
  \]
  that is
  \[
    | m-n | + | m-4n | = | m-2n | + | m-3n |.
  \]
  Therefore, $ \abseq $ is a member of $ \coexcaffine $.
  \qed\end{proof}

\begin{corollary}
  $ \coexcstoc  = \overline{\nqal} = \overline{\naffine} \subsetneq \coexcaffine $
  and
  $ \excstoc = \nqal = \naffine \subsetneq \excaffine $.
\end{corollary}
Then, similar to PFAs, if we change the cutpoint in $ (0,1) $, the class $ \excaffine $ does not change, but when setting to 0 or 1, we obtain $ \naffine $, a proper set of $ \excaffine $.

\section{Concluding remarks}
\label{sec:conclusion}

We introduce affine computation as a generalization of probabilistic computation by allowing to use negative ``probabilities''. After giving the basics of the new system, we define affine finite automaton and compare it with probabilistic and quantum finite automata. We show that our new automaton model is more powerful than the probabilistic and quantum ones in bounded- and unbounded-error language recognitions and equivalent to quantum one in nondeterministic language recognition mode. Moreover, we show that exclusive affine languages form a superset of exclusive quantum and stochastic languages. These are only the initial results. Recently, new results regarding computational power and succinctness of AfAs are obtained in \cite{VilY16A,BMY16A}. We believe that the further investigations on the affine computational models can provide new insights on using negative transition values.

\section{Acknowledgements} We thank Marcos Villagra for his very helpful comments. We also thank the anonymous reviewers for their very helpful comments.

\bibliographystyle{plain}
\bibliography{tcs}

\end{document}